\newcommand{\esa}[2]{#2}{} 
\documentclass[11pt,letterpaper]{article}

\pdfoutput=1 
\usepackage{amsfonts,amsmath,amssymb}
\usepackage{amsthm}
\usepackage{color,hyperref}
\usepackage{graphicx,psfrag}
\usepackage{paralist}
\usepackage{url}

\newcommand{\alloc}{\mathcal{A}}
\newcommand{\ctr}{\alpha}

\newcommand{\numitems}{m}
\newcommand{\numtotal}{n}
\newcommand{\MPP}{MPP_{CA}}
\newcommand{\VCG}{VCG_{CA}}
\newcommand{\eff}{\mathcal{E}}

\newcommand{\sumbids}[2]{\sum_{#1\in #2}b_{#1}}
\newcommand{\sumvals}[2]{\sum_{#1\in #2}v_{#1}}

\newcommand{\eg}{e.g., }

\newcommand{\junk}[1]{}
\newcommand{\XXX}[1]{}
\newcommand{\setmpp}{\mu}
\newcommand{\setpgsp}{\delta}
\newcommand{\settwompp}[1]{\mu_{#1}}
\newcommand{\setopt}{\sigma}
\newcommand{\settwoopt}[1]{\sigma_{#1}}
\newcommand{\anotb}[2]{#1 \backslash #2}
\newtheorem{example}{Example}[section]
\newtheorem{lemma}{Lemma}[section]
\newtheorem{theorem}{Theorem}[section]

\begin{document}

\title{Analyses of Cardinal Auctions}
\author{Mangesh Gupte\thanks{This work was done while the author was graduate student at Rutgers University.
}\\\
	Google Inc\\
	\href{mailto:mangesh@cs.rutgers.edu}{mangesh@cs.rutgers.edu}
\and
Darja Krushevskaja\\
	Rutgers University\\
	\href{mailto:darja@cs.rutgers.edu}{darja@cs.rutgers.edu}
\and	
S. Muthukrishnan\\
	Rutgers University\\
	\href{mailto:muthu@cs.rutgers.edu}{muthu@cs.rutgers.edu}
}

\maketitle

\begin{abstract}

We  study  {\em cardinal auctions} for selling multiple copies of a good,  in which bidders specify not only their bid or how much they are ready to pay for the good, but also a cardinality constraint on the number of copies that will be sold via the auction.
We perform first known \emph{Price of Anarchy} type analyses with detailed comparison of the classical
Vickrey-Clarke-Groves  (VCG) auction and one based on minimum pay property (MPP) which is similar to Generalized Second Price auction commonly used in sponsored search. 
Without cardinality constraints, MPP has the same efficiency  (total value to bidders) and at least as much revenue (total income to the auctioneer) as VCG; this also holds for certain other generalizations of MPP (\eg prefix constrained auctions, as we show here). In contrast, our main results are that, with cardinality constraints, 
\begin{inparaenum}[(a)]
\item equilibrium efficiency of MPP is 1/2 of that of VCG 
and this factor is tight, and 
\item in equilibrium MPP may collect as little as 1/2 the revenue of VCG. 
\end{inparaenum}
These aspects arise because in presence of cardinality constraints, more strategies are available to bidders in MPP, including bidding above their value, and this makes analyses nontrivial.

\end{abstract}

\section{Introduction}

Assume we have $\numtotal$ bidders and $\numitems$ identical items to sell via an auction, s.t. $\numtotal\ge \numitems$. What auction is suitable? In particular, there are three decisions to be made:
 \begin{itemize}
\item how many items to sell: $k^{*}$, 
\item how to allocate $k^{*}$ items: $a(.)$ and 
\item how to price each of them: $p(.)$. 
\end{itemize}
We consider the case of {\em negative externality} when the number of bidders who win and are allotted the item affects the value of the item to each of the winners. In particular, each bidder is not interested in a copy if the number of copies eventually allocated exceeds her threshold. {\em Cardinal auctions} explicitly incorporate this externality into the bidding language. In this paper we study efficiency and revenue tradeoffs for such auctions.

\section{Model}

In cardinal auctions, there are $n$ buyers competing for at most $m \le n$ identical copies of an
item in the auction. Each buyer wants to buy exactly one copy and has two
{\em private} numbers $v_i$ and $k_i$. Auctioneer has no prior information about values of buyers.  The utility $u_{i}(v_{i}, k_{i})$ that bidder $i$ derives from the auction is 
\begin{eqnarray*}
u_{i}(v_{i}, k_{i}) = \left\{
\begin{array}{l l}
     x_{i}v_i-p_{i}& \quad \text{if number of copies sold is less than $k_{i}$}\\
    -\infty& \quad \text{otherwise}\\
  \end{array} \right.
\end{eqnarray*}
where $x_{i} \in \{0,1\}$ is indicator variable that shows whether $i$ was allotted a copy or not, and $p_i$ is the price at which $i$ obtains it. 

Buyers express their preferences through \emph{2 dimensional bid} $(b_{i}, l_{i})$ where $b_{i}$ is the maximum amount buyer $i$ is willing to
pay if at most $l_{i}$ copies are allocated. Note, that $b_{i}$ may differ from $v_{i}$, and $l_{i}$ from $k_{i}$. 
Once the auctioneer gathers all the
bids she has to decide on optimal number of copies $k^*$,  allocation of $k^*$ copies according to function $a(.)$ and payments according to pricing function $p(.)$.  In mechanisms we will consider,  no bidder $i$  will be a winner if $l_i < k^*$. 

\junk{
We consider mechanism that satisfy the following properties: 
\begin{itemize}
\item \emph{observe cardinality preferences:} mechanism never assigns a bidder into allocation of size greater than $l_{i}$.
\item \emph{individual rationality: } an agent's utility from participating in the auction is non-negative. That implies that mechanism never chargers more than bid: $b_{i} \ge p_{i}$.
\end{itemize}
}

\medskip
\noindent
{\bf Motivating Scenarios.}
An important motivation arises in auctions for online advertisements (ads). 
Consider {\em display ads}, or visual ads, on a webpage. Advertisers whose ad is shown on the page compete for attention of the viewers.
Clearly, the {\em number} of ads shown is an important feature, e.g., 
publishers recognize that showing fewer ads helps\footnote{\url{http://www.technologyreview.com/web/25827/?a=f}}.
Currently, this cardinality is largely determined by the publisher of the web page, who may choose to make it exclusive showing only one ad, but in many cases mixes several. They choose the  number of ads on a page based on variety of techniques from machine learning to user studies, 
esthetics of UI design and revenue maximization. This approach does not let advertisers influence how many ads appear with their own; hence, they bid depending on the average of their values over the possible number of ads that might appear on that page. This induces inefficiencies and potential revenue loss.  Cardinal auctions are an alternative. They let advertisers explicitly  specify how many other advertisers may appear with their ad on a given page. 

Cardinal auctions are also suitable in a variety of other instances:  
 \begin{itemize}
 \item Say we can produce a collectors item such as a signed copy of an album or a book. The more exclusive the copy is, the more valuable it is to the possessor.  How many copies shall we produce? While traditionally this is determined by estimating the demand function, one can imagine an auction-based method, where bidders can specify in some way the value of the item to them as a function of how many copies are made and sold. 
 \item
Consider a situation that arises in a data exchange such as BlueKai\footnote{\url{http://www.bluekai.com/}} where certain pertinent data about a user is sold for ads targeting. The data may be sold to any number of advertisers for targeting, but in some cases, the more the information is shared, the less value it gives to the advertisers. Hence, when data is sold via auction, advertisers may wish to be able to influence how many of others get access to the data.  
\end{itemize}


\subsection{Auctions}
Allocation $\alloc$ is the set of $k^*$ winners who obtain a copy. 
We consider set of \emph{feasible} allocations: allocation $\alloc$ is {\em feasible} if $\{l_{i} \ge  |\alloc|,  k_i \geq |\alloc| : \forall i \in \alloc \}$. The total {\em efficiency}  $\eff_{\alloc}$ of allocation $\alloc$  is the sum of values of allotted bidders or $\sum_{i\in \alloc} v_i$. 

There are two natural auctions to consider. 

\vspace{3mm}
\noindent 
$\mathbf{\VCG}$. $\VCG$ is a straightforward extension of the standard
Vickrey-Clarke-Groves (VCG) auction~\cite{V,C,G}. 
$\VCG$ is  \emph{truthful}, i.e., bidder's bid their true valuations. This is a well-known property of VCG
mechanism. Thus we can assume that bidders submit their  bids as $(v_{i}, k_{i})$. $\VCG$ chooses the feasible
allocation that maximizes the total sum of values:  $\eff^{*} = \max_{\alloc}\sum_{i\in\alloc}v_{i}$.
Let $\eff_{-i}$ be highest efficiency achievable without bidder $i$, then the
price for bidder $i$ is 
\begin{eqnarray}
\label{eq:vcgpayment}
p^{VCG}_{i} = \eff_{-i} -\eff^{*}+ v_{i}
\end{eqnarray}

$\VCG$, as  generally known, can have low revenue. Furthermore, with cardinality constraints its outcomes are not \emph{envy-free}: losing bidder would agree to purchase a copy for the price higher than what is being asked from winners. For more intuition, consider the following example:

\begin{example}
There are 3 bidders with true valuations: $A=(100, 1)$, $B = (90,2)$ and $C=(80,2)$. For this setting $\VCG$ will identify winning allocation with bidders $B$ and $C$ in it since the total efficiency is $90+80 = 170>100$). It will charge $B$ amount $p_{B}= 100 - 170 + 90 = 20$ and $C$ amount $p_{C} = 100 - 170 + 80 = 10$. Thus total collected payment is 30. However, bidder $A$ will envy this low payment of 30. 
\hfill $\Box$
\end{example}

\medskip
\noindent
$\mathbf{\MPP}$.  $\MPP$ was introduced in~\cite{muthukrishnan2009bidding}  and is based on the \emph{minimum pay property} for the outcome: auction requires every buyer to pay no more than what she would have bid, if
she knew all other bids, to get the exact same assignment she got. 
To calculate prices, let winning allocation be $\alloc^{*}$ and $\alloc_{2}$ the allocation that gives second highest sum of bids after $\alloc^{*}$.  Let winning bids in $\alloc^*$
 be sorted top-down in decreasing order of bids, $b_1 \geq b_2 \geq \dots$.
 The $i$th winner pays price
 \begin{eqnarray*}
 p^{MPP}_{i} = \max \{\sumbids{j}{\alloc_{2}} - \sumbids{j}{\alloc^{*}} + b_{i} , b_{i+1} \}
 \end{eqnarray*}
 The price consists of two components. The first term is the minimum amount $i$ needs to bid to ensure that the allocation $\alloc^*$ is the winner, and the second term is the minimum bid to get above the $i+1$'st largest bid. The overall price is the maximum over both. 
 
 MPP auction is inspired by  \emph{Generalized Second Price} (GSP) auction used by 
 many popular search engines~\cite{edelman2005internet,V,aggarwal2006truthful} to determine placement of advertisements (ads) on the page. 
In $GSP$ there are $n$ advertisers bidding for $m$ advertisement slots. Each slot $i$ has associated \emph{click through rate} (CTR) with it, or probability of being clicked, denoted by $\alpha_{i}\in (0,1)$. Slots are ordered in decreasing order of CTR's:  $\alpha_{i} > \alpha_{j}$ for $i<j$. Advertiser $i$ has private valuation $v_{i}$, which expresses the value of getting a click. 
To participate in the auction advertiser submits bid $b_{i}$ that indicates maximum payment she is willing to make. 
Auctioneer receives all bids, and assigns advertisers to slots in decreasing order of  their bids. For convenience, let us renumber advertisers in decreasing order of their bids, then, advertiser $i$ is assigned to slot $i$ with CTR $\alpha_{i}$. Payment of advertiser $i$ is $p_{i} = \frac{\alpha_{i+1}}{ \alpha_{i}}b_{i+1}$, and is charged only if the ad is clicked. 
 
 $\MPP$ naturally generalizes $GSP$ auction and in absence of cardinality constraints is the special case of $GSP$ without click through rates, or when $\alpha_{i}=1$ for all $i$.   

\subsubsection{Analysis of Auctions}
Unlike $\VCG$, $\MPP$ is not truthful. While bidder cannot benefit from misreporting her cardinality preference\footnote{we will show this in a separate argument later on}, she can improve her utility by reporting $b_{i} \ne v_{i}$. Consider the following example: 

\begin{example}
Consider 3 bidders with their true valuations: $A =(100,1)$, $B = (80,2)$ and $C=(70,2)$. Auctioneer runs $\MPP$. If auctioneer receives truthful bids, then she will choose allocation of 2 bidders: $(B, C)$, and charge them 70 and 20, respectively. Utility of bidder $B$ is $u_{B} = 80 - 70 = 10$. However, bidder $B$ can improve it by lowering her bid to $(40, 2)$. Then, allocation is the same (bidders $B$ and $C$), but payments are different: payment of $B$ is 30, and payment of $C$ is  60. Now, utility of $B$ is $u'_{B} = 80 - 30 = 50>10$. Hence, bidder $B$ benefits from submitting untruthful bid $b_{B} < v_{B}$.
\hfill $\Box$
\end{example}

$\MPP$ can have many outcomes,  
we consider only bid vectors that are in \emph{Nash equilibria}, that is, for every bidder $i \in \alloc$ the following inequalities hold:
\begin{eqnarray*}
 v_{i} - p_{i, \alloc} \ge v_{i} - p_{j, \alloc'} & & \forall \alloc' \in \mathcal{F}, \alloc \ne \alloc'\\
v_{i} - p_{i,\alloc} \ge v_{i} - p_{j, \alloc} & & \forall j\ne i
\end{eqnarray*}
There is a set of Nash equilibria efficiencies of  $\Sigma_1,\Sigma_2, \dots$. 
Define $\eff_{min}=\min_{i} \Sigma_i$ and $\eff_{max}=\max_i \Sigma_i$.  Then, \emph{price of anarchy} is defined as 
$PoA(.) = \frac{\eff_{max}}{\eff_{min}}$. Observe, that $\eff_{max} = \eff^{*}$. Thus, $PoA(\VCG) = 1$, furthermore in order to evaluate performance of $\MPP$, one can use $\VCG$ as a benchmark, and compare efficiency of $\VCG$ outcome with that of the worst outcome of $\MPP$: 
$$PoA(\MPP) = \frac{\eff^{*}}{\eff_{min}}.$$
 
To analyze revenue we use $\VCG$ as the benchmark for consistency, and compare it to the revenue collected by $\MPP$.

Bidders are strategic and their goal is to maximize their utility. Their behavior, or \emph{strategy}, is determined by mechanism. Strategy $s$ is \emph{weakly dominant} if regardless of what other bidders do, strategy $s$ gets a player utility that is at least as high as utility obtained by playing any other strategy.  Strategy $s$ is \emph{(strictly) dominant} if utility of playing strategy $s$ is strictly larger than playing any other strategy, regardless what other bidders do.  We consider two types of bidders:
\begin{itemize}
\item \emph{Conservative bidders} do not bid over their value, i.e. $b_{i} \le v_{i}$, hence they do not risk paying more than their true valuation and getting negative utility.
\item \emph{Rational bidders} can bid above their true valuation $v_{i}$ in equilibria iff the payment $p_{i}$ does not exceed $v_{i}$. Equilibria that contain such bids are fragile, because bidder can get negative utility if some other bidder changes her bid. Such equilibria help us explore the properties of possible outcomes. 
\end{itemize}

\subsection{Our Results} 
We perform first known analyses of cardinal auctions, and compare efficiency and revenue of  $\VCG$ vs $\MPP$ in equilibrium. 
\junk{
\begin{description}
\item[Efficiency] We were able to give tight bounds for both conservative and rational bidders. For conservative bidders we show price of anarchy of $\MPP$ is 1. However, hoe rational bidders it is 2.
\item[Revenue] For conservative bidders, we show that revenue of $\MPP$ is
  always larger than revenue of $\VCG$.  For rational bidders, we show that
  revenue of $\MPP$ may be even $1/2$ that of $\VCG$.
\end{description} 

Many more variations of cardinal auctions are of interest, we discuss these in
the concluding remarks for future work. In this paper we briefly discuss following modification of the model: bidders want to have handle on how early they will adopt the product: e.g. bidder $i$ is interested in obtaining the good only if she is the first receiver. The model is a special case the model discussed in~\cite{aggarwal2007bidding}, in our modification all items are identical.~\cite{aggarwal2007bidding} discusses model and shows that there exist an equilibria that has allocation and pricing identical to $VCG$. 
}

\begin{itemize}
\item {\em (Efficiency)} 
We show that PoA is $1$ for conservative bidders. For rational bidders, 
we show that PoA is $2$ and this  is optimal. 

\noindent
As noted before, without cardinal constraints $l_i$, $\MPP$ becomes GSP without click through rates. In that case, it is a dominant strategy for bidders to be conservative, and  PoA is $1$. Further, we show that, even with a slightly different bidding language of {\em prefix constraints} (defined later precisely), it is still (weakly) dominant strategy for bidders to be conservative, and PoA of prefix auctions is still $1$. It is interesting that with cardinality constraints, there is provable loss of efficiency in equilibrium. 

\item {\em (Revenue)}
We show that for conservative bidders revenue of $\MPP$ is always at least that of $\VCG$, and for rational bidders, revenue of $\MPP$ may be only $1/2$ of that of $\VCG$. 

\noindent
In contrast, without cardinality constraints, $\MPP$ has larger revenue than $\VCG$, no matter the nature of bidders. 
\end{itemize}

In both analyses of efficiency and revenue, the central technical challenge is that $\MPP$ pricing has two components, first ensures that the eventual allocation is ``better'' than others with fewer, or more items, and the other is the impact of the bidder in the chosen allocation. The role of position component of price of $\MPP$ has been studied extensively in analyses of GSP~\cite{leme2010pure,lucier2011gsp}.  
 The cardinality component of the pricing of $\MPP$ allows bidders to bid over their true valuation, this  induces the nontrivial PoA properties, complicates the analyses, and is novel. 
\subsection{Related Work}
The work closest to ours is~\cite{ghosh2010expressive,jerath2011exclusive,muthukrishnan2009bidding}.
In~\cite{muthukrishnan2009bidding} authors introduce auction with cardinal externality, formulate and motivate pricing, and give efficient algorithms for calculating the allocation and pricing.  However, they do not perform PoA type analysis we do here. 

\cite{ghosh2010expressive,jerath2011exclusive} introduce auction with negative externality: valuation of bidder $i$ is a pair $(v_{i}^{E}, v_{i}^{M})$, where $v_{i}^{E}$ is value of the bidder for being {\em exclusively}  allocated, and $v_{i}^{M}$ is valuation if bidder $i$ for being allocated among with other bidders.  The auction determines allocation type (exclusive vs non exclusive), set of bidders allotted, and prices at which they get the item. The authors consider various pricing schemes including variations of VCG and MPP and provide analysis of efficiency and revenue in equilibrium. Our work differs from~\cite{ghosh2010expressive,jerath2011exclusive} in the bidding language. The bidding language 
in~\cite{ghosh2010expressive,jerath2011exclusive} can not specify preference for the number of copies being sold, while our bidding language does not allow to specify more than a single bid. Thus, the bidding languages are incomparable. We believe both  languages are natural and interesting. 
 
Another relevant auction was presented in~\cite{aggarwal2007bidding}, where authors consider the auction for the ordered set of items, in which bidders specify the largest prefix of allocation in which they want to participate: valuation of the bidder is $(v_{i}, k_{i})$, where bidder $i$ has valuation $v_{i}$ only if she is in top $k_{i}$ allotted bidders. In their model bidders have no influence on valuations of each-other or the size of final allocation. Authors present two auctions: one based on VCG and the other based on MPP. The paper shows that MPP based auction can achieve efficiency of VCG, however do not present PoA type results. In this paper we present analysis of efficiency and revenue for this model and contrast it with our main results of efficiency and revenue for cardinal auctions.


In general, there are limited PoA style analyses of auctions. 
Closely related to cardinal auctions are sponsored search auctions where each item has an associated click through rate.
A recent paper performs PoA style analysis in this setting: 
in ~\cite{leme2010pure} authors show that it is weakly dominant strategy for bidders to bid conservatively and  PoA of $\frac{1+\sqrt{5}}{2}$ for Nash equilibria of GSP with conservative bidders. 

There is extensive body of literature on \emph{multi-unit auctions}: auctioneer has $m$ items to sell among $n$ buyers. The problem was studied at least since~\cite{shapley1971assignment}, however much of that work does not accommodate externalities, and considers the number of items to sell $k$ as an input to the model.
A more recent phenomena is  auctions with unlimited supply, or  digital good 
auctions~\cite{goldberg2003envy,goldberg2006comp}. Here, auctioneer has unlimited supply of the good to sell, she needs to find optimal price $p^{*}$ that determines both allocation and uniform price. However  buyer valuation $v_{i}$'s  are independent of the outcome of the allocation, unlike what we study here.

Externalities have been studied extensively in Economics and Computer Science. 
~\cite{haghpanah2011optimal} considers model where buyers experience  positive externalities once sufficient number of their friends are also allocated the copy of the item. In~\cite{jehiel1996not} authors consider situation when winning buyer subjects other participants to negative externality.
~\cite{aseff2008optimal} studies the problem of allocating a pair of goods among group of sellers who have pairwise externalities to each other. 

~\cite{pei2011working} considers auction for sharable goods. In their setting,  valuation of the bidder is multidimensional and represents value of strict ownership, sharing value, and that of no allocation. They investigate unique perfect Bayesian equilibrium. All of these works 
consider externality as an input to the problem or model value of a bidder as some function of number of items allocated. All bidders are treated equally and no bidder has control over allocation in which they participate. This is different from
cardinal auctions we study here where the bidding language explicitly allows bidders to specify the cardinal constraints. 

At the highest level, it will be of interest  to study auctions for the model in which buyer can explicitly define her valuation for each possible allocation size and the set of different identities of co-winners. However, such exponential input can not be processed or analyzed efficiently, and even in restricted cases, at least as hard as various multidimensional auctions.


In practical terms,
within the context of advertising, there were machine learning approaches to estimating the optimal number of ad slots on the page. In~\cite{schroedl2010generalized} authors consider several utility functions that incorporate user experience in order learn number of slots that optimizes utility unction.~\cite{broder2008swing} formulates a binary problem: given a set of relevant ads the goal is to decide whether system benefits from showing ads or should it not show the ads and by that benefit in a long run.  These are interesting directions, different from our approach which relies on the bidders to address the issue (at least in the short run).

\section{Preliminary Observations}

Our first observation is regarding truthfulness of the cardinal constraint. 
\begin{lemma}
\label{lemma:truth}
In $\MPP$,  bidders truthfully reveal their private $k_i$'s, that is $l_i=k_i$ for all $i$ in Nash equilibria of $\MPP$ auction.
\end{lemma}
\begin{proof}
Consider any bidder $i$ who participates in winning allocation of size $k^{*}$. For contradiction let their revealed $l_i > k_i$. 
There are two possibilities: 
\begin{inparaenum}[(1)]
\item $k^{*} \leq k_i$, then bidder $i$ would have identical utility by reporting 
$k_i$ instead of $l_i$;
\item $k^{*} > k_i$, bidder now has utility
$-\infty$ which is worse than utility from reporting $k_i$ truthfully.
\end{inparaenum}
 Hence, bidder does not have incentive to submit $l_i > k_i$.

Consider same bidder $i$, and say for contradiction that their revealed $l_i < k_i$. 
Again, there are two cases:
\begin{inparaenum}[(1)]
\item $k^{*} \leq l_i$, then bidder $i$ would have identical utility by reporting 
$k_i$ instead of $l_i$;  
\item with some positive probability  $k_{i} \ge k^{*}> l_i$, bidder now has utility
$0$ which is worse than utility from reporting $k_i$ truthfully. 
\end{inparaenum}
Hence, bidder does not have incentive to submit $l_i < k_i$.
	\esa{	
	\hfill $\Box$
 	}{}
\end{proof}

Our second observation concerns bidding behavior of \emph{losing bidders}. 
There exist bids which are in equilibrium, but the efficiency is bounded away from the maximum achievable efficiency by an arbitrarily large factor. Consider the case of three bidders $(100,1), (75,2), (75,2)$, who make the following bids: $(100,1)$, $(1,2)$, $(1,2)$. This set of bids forms Nash equilibrium since no bidder by herself has an incentive to change her bid. However, the efficiency of the resulting allocation is $100$, compared to the optimum efficiency of $150$. This is due that losing bidders can arbitrarily shade their bids.  This is a common problem, previously faced in~\cite{ghosh2010expressive}. Like~\cite{ghosh2010expressive}, we will henceforth assume that losing bidders bid their true valuation. 

\section{Efficiency Analyses}

\begin{theorem}
With conservative bidders, $\MPP$'s allocation has the same total value  as $\VCG$ in Nash equilibrium. 
\label{kstarlemma}
\end{theorem}
\begin{proof}
Let $\setopt=\settwoopt{VCG}$ be the set of winning bidders that maximizes efficiency and $\setmpp$ be the set of winning bidders under $\MPP$ in Nash equilibria. 
Let $\{b_i | i \in \setmpp\}$ be a set of equilibrium bids under $\MPP$. 
Since $\MPP$  chooses the set of bidders who maximizes total sum of bids, then it must be true that $
\sumbids{j}{\setmpp} \geq  \sumbids{j}{\alloc}
$.
Since $\setopt$ is feasible,  $\sumbids{j}{\setmpp}  \geq   \sumbids{j}{\setopt}$.  Hence
  \begin{align}
    & &\sum_{i \in \anotb{\setmpp}{\setopt}} b_i + \sum_{i \in \setmpp \cap \setopt} b_i
    & \geq   \sum_{j \in \anotb{\sigma}{\setmpp}} b_j + \sum_{j \in \setmpp \cap \setopt} b_j \nonumber \\
    &\implies& \sumbids{i}{\anotb{\setmpp}{\setopt}} & \geq  \sumbids{j}{\anotb{\setopt}{\setmpp}}  \nonumber \\
    &\implies& \sumvals{i}{\anotb{\setmpp}{\setopt}} & \geq \sumvals{j}{\anotb{\setopt}{\setmpp}} \label{eq:assumption}\\
    &\implies& \sumvals{i}{\setmpp}  & \geq  \sumvals{j}{\setopt} \nonumber 
  \end{align}
In (\ref{eq:assumption}) we use the assumption on the right hand side that losers bid at least their true valuations and on the left hand side that bidders are conservative.  Then the only possibility is that total value of $\setmpp$ equals that of  $\setopt$.
\esa{
\hfill $\Box$
}{}
\end{proof}

It follows that the PoA of $\MPP$ is $1$ for conservative bidders. 

\begin{theorem}
For rational bidders, PoA of $\MPP$ is $2$ and this is tight.
\end{theorem}
\begin{proof}
Let $\setmpp$ and $\setopt$ denote the set of bidders chosen by the allocation of $\MPP$ and $\VCG$ respectively. Let $\settwompp{2}$ denote set of bidders who belong to second best allocation of $\MPP$.

If $\setopt = \setmpp$,  then the efficiency is 1 and we are done. Otherwise,
\begin{align}
& & \setmpp = \sumbids{i}{\setmpp} & > \sumbids{j}{\setopt} \nonumber \\
&\implies& \sum_{i \in \anotb{\setmpp}{\setopt}} b_i + \sum_{i \in\setmpp \cap \setopt}b_i
& >  \sum_{j \in \anotb{\setopt}{\setmpp}}b_j + \sum_{i \in \setmpp \cap \setopt}b_i \nonumber \\
&\implies& \sum_{i\in\anotb{\setmpp}{\setopt}} b_i & > \sum_{j\in\anotb{\setopt}{\setmpp}} b_j \nonumber \\
&\implies& \sum_{i \in\anotb{\setmpp}{\setopt}}b_i & > \sum_{j \in\anotb{\setopt}{\setmpp}} v_j \label{eq:trans}
\end{align}
where to get~(\ref{eq:trans}) we use assumption that losing bidders bid at least their value. Remainder of the proof deviates from the conservative bidder case as we can not bound the left hand side for rational like we did with conservative bidders. 

Without loss of generality, assume that the bidders in $\anotb{\setmpp}{\setopt} = \{b_1, b_2, \ldots, b_k \}$ are ordered in non-increasing order of bids,  i.e., $b_1 \geq b_2 \geq \dots \geq b_k$. 
To lowerbound the payment of the highest bidder, we start by working on one of the components of the pricing:
\begin{eqnarray}
 \sumbids{i}{\settwompp{2}} - \sumbids{i}{\setmpp} + b_{1} &\ge& \sumbids{i}{\setopt} - \sumbids{i}{\setmpp} + b_{1} \nonumber \\
&\ge &\sumvals{i}{\anotb{\setopt}{\setmpp}}  - \sum_{2\leq i \leq k}b_i + b_{1}  \label{eq:pricelower}
\end{eqnarray}
where we get the first term of~(\ref{eq:pricelower}) from Eq.~\ref{eq:trans}. Notice, that  if highest bidder $i$ belongs to $\setopt$, then there is at least one bidder in $\setmpp$ who pays more then her value and gets negative utility. Hence, for allocation to be in Nash equilibria highest bidder $i$ must belong to  $\anotb{\setmpp}{\setopt}$, and we can exclude highest bidder from second term of~(\ref{eq:pricelower}).
Hence,
$$p(b_1) \geq \max\left\{ b_2 , \sum_{i \in \anotb{\setopt}{\setmpp}} v_i -\sum_{2\leq i \leq k}b_i\right\}$$ 
For other bidders, we bound $\MPP$ payment  by the $p(b_i) \geq b_{i+1}$. Using these, we get a lower bound on the total revenue of $\MPP$ as follows:
 \begin{eqnarray*}
 \sum_{1 \leq i \leq k} p(b_i) &=& p(b_1) + \sum_{2 \leq i \leq k} p(b_i) \\
 &\geq& \max \left\{b_2 , \sum_{i \in\anotb{\setopt}{\setmpp}} v_i - \sum_{2 \leq i \leq k} b_i \right\} +  \sum_{2 \leq i \leq k} b_{i+1} \\
 &\geq& \max \left\{b_2 , \sum_{i \in \anotb{\setopt}{\setmpp}} v_i - \sum_{2 \leq i  \leq k} b_i + \sum_{3 \leq i \leq k} b_i  \right\} \\
  &= & \max \left\{ b_2 , \sum_{i \in \anotb{\setopt}{\setmpp}}v_i - b_2 \right\} 
 \geq  \frac{1}{2} \sum_{i \in \anotb{\setopt}{\setmpp}} v_i
 \end{eqnarray*}
 Since, the bidders are rational,
 $$\sum_{i \in\anotb{\setmpp}{\setopt}} v_i \geq \sum_{i \in\anotb{\setmpp}{\setopt}}p_{i}= \sum_{1 \leq i \leq k} p(b_i)$$
 and chaining with the previous equation, we get that
$ \sum_{i \in\anotb{\setmpp}{\setopt}} v_i \geq \frac{1}{2} \sum_{v_i \in \anotb{\setopt}{\setmpp}} v_i
$.

\medskip
\noindent
{\bf Tightness.}
To see that the bound is tight consider three bidders with the following valuations: $(100,1)$, $(50,2)$, $(\epsilon,2)$ and the bids they place are $(100,1)$, $(100,2)$, $(50,2)$. Bids form NE, and none of the bidders can improve her utility acting on her own.  $\MPP$ will allocate bidders $(2, 3)$ and charge them $50$ and $0$ respectively. PoA is then $\frac{50 + \epsilon}{100} \approx \frac{1}{2}$.
 \esa{
 \hfill$\Box$
 }{}
 \end{proof}

\paragraph{Contrasts with other auctions} 
As mentioned earlier, without cardinality constraints $\MPP$ becomes $GSP$ auction without click-through rates. It is known, that in that case PoA of $GSP$ is 1. 
To highlight our result further, we consider PoA of \emph{prefix auctions}~\cite{aggarwal2007bidding} and show that it is also 1. 
\junk{
In prefix auctions $n$ bidders are bidding on the ordered set of $k$ goods. Bidder $i$ want's at most 1 item and has valuation $v_{i}$ only for goods that are located in prefix of size $k_{i}$. In ~\cite{aggarwal2007bidding} authors consider role of this externality in \emph{ad auctions}, In prefix ad auctions there are $k$ ordered ad slots, each slot $i$ has associated CTR $\ctr_{i}\in[0,1]$. Bidder $i$ submits two dimensional bid $(b_{i}, l_{i})$ and extracts utility $u_{i} = \ctr_{i}(v_{i} - p_{i})$ if $i\le k$. Authors present two different auction one based on $VCG$ and the other on $MPP$. They show that $MPP$ based auction has an equilibrium that matches allocation and prices of $VCG$ based auction. Authors do not perform price of anarchy types analysis. 

In this paper, we consider first known price of anarchy of prefix auction for identical ordered goods. 
}

The model is as follows. There are $n$ ordered identical items to sell. There are $m$ bidders,  each bidder $i$ has two private values  $v_i$ and $k_i$. Utility $u_i$ of bidder $i$ is $v_i-p_{i}$ if she obtains \emph{any} of the first $k_i$ copies, and it is $-\infty$ otherwise.  Notice that now, 
bidder $i$ has positive utility even if more than $k_i$ copies are auctioned. 
In the auction, each bidder  $i$ submits a pair $(b_{i}, l_{i})$: $b_i$ is the maximum they are willing to pay, if they are allotted one of the first  $l_i$ copies. 

We consider two different auctions: $pVCG$ and $pGSP$:\\
$\mathbf{pVCG}.$ $pVCG$  is extension of $VCG$ and  is truthful. In the auction bidders submit their true valuations $(v_{i}, k_{i})$ to auctioneer, upon recieval of bids auctioneer creates feasible allocation that maximizes total efficiency and calculates payments using Eq.~\ref{eq:vcgpayment}.

\medskip
\noindent
$\mathbf{pGSP}.$ $pGSP$ is iterative second price (SP) auction from first copy of the item to the last, where for each item, we run a SP auction among bidders who are not yet assigned a copy but who still have nonnegative utility from obtaining one. 

Notice, that bidder cannot benefit by submitting bid $b_{i}>v_{i}$, hence it is weakly dominant strategy for bidders to bid {\em conservatively}. Similarly to $\MPP$ it is dominant strategy for the bidder to report her preference $k_{i}$ truthfully. The argument is identical to Lemma~\ref{lemma:truth}.
 
Unlike second price auction $pGSP$ is not truthful. 
Thus, similarly to $\MPP$, we analyze bid vectors that are in Nash equilibria. For $pGSP$, Nash equilibrium is defined as follows. For each $i$, 
\begin{eqnarray*}
v_{i} - p_{i} \ge v_{i} - p_{j} \ \forall j \ne i \text { and } j \le \min\{k, k_{i}\}
\end{eqnarray*}

Here we give PoA results for prefix auctions without click-through rates.


\begin{theorem}
\label{t:eff}
$PoA$ of $pGSP$ is $1$.
\end{theorem}

\begin{proof}
Let $\setopt$ be set of bidders allocated by $pVCG$ and $\setpgsp$ be set of bidders allocated by $pGSP$ in Nash equilibria. 
Then, similarly to Eq.~\ref{eq:trans} we have
\begin{eqnarray*}
\sum_{i \in \anotb{\setopt}{\setpgsp}} v_{i} &\ge& \sum_{j \in \anotb{\setpgsp}{\setopt}} v_{j} 
\end{eqnarray*}

It is possible only if
\begin{eqnarray}
\exists i\in \anotb{\setopt}{\setpgsp} \ \forall j \in \anotb{\setpgsp}{\setopt} \mbox{ s.t. } v_{i} > v_{j} \label{eq:cond}
\end{eqnarray}

Assume, it is not the case. Then,
$\exists j\in  \anotb{\setpgsp}{\setopt} \ \forall i \in \anotb{\setopt}{\setpgsp} \mbox{ s.t. } v_{j} > v_{i}$.
However, this is possible only if bidder $v_{j}$ cannot replace any of bidders $i\in \anotb{\setopt}{\setpgsp}$, otherwise $v_{j}$ would improve $\eff^{*}$. This, in turn, is possible if and only if $|\{i | i\in \anotb{\setopt}{\setpgsp}\}| = 0$ that would imply that $|\setopt| < |\setpgsp|$. However, it gives a contradiction, as efficiency of $\eff^{*}$ could be improved by adding $v_{j}$ to it.

If~(\ref{eq:cond}) is true, then there must be a  losing bidder $l$ who can raise her bid, enter the allocation and as the result improve her utility. That gives a contradiction. Thus, total values of  $\setopt $ and $\setpgsp$ are identical.
\esa{
 \hfill$\Box$
 }{}
\end{proof}
\section{Revenue Analyses}
\junk{
Auction is conducted by the auctioneer, sure gain of the auctioneer is revenue. VCG auctions are considered to collect low revenue, MPP auctions by construction aim to improve this characteristic. Let's consider revenue of $\MPP$ with respect to $\VCG$ to see whether $\MPP$ improves on the revenue of $\VCG$ as expected. 
}

Let $Rev(X)$ be revenue generated by mechanism $X$. We show two results. 

\begin{theorem}
With conservative bidders, $Rev(\MPP) \ge Rev(\VCG)$.
\label{teorem:conservativerevenue}
\end{theorem}
\begin{proof}
Let $\setopt$ be the allocation with maximum total value (hence, the value attained by $\VCG$),  $\setmpp$ be the allocation of
$\MPP$ in equilibrium,  $\settwompp{2}$ be the set of bidders who participate in second best allocation of $\MPP$ and 
$\settwoopt{-i}$ the set of bidders that gives the largest total value allocation when bidder $i$ is not present.

Consider payments of each bidder $i \in \setopt \cap \setmpp$ under $\VCG$ and $\MPP$:
\begin{eqnarray*}
p^{VCG}_{i} = \sum_{j \in \settwoopt{-i}} v_j - \sum_{j\in\setopt} v_j + v_i
\end{eqnarray*}
\begin{eqnarray*}
p^{MPP}_{i} &=& \max \{ b_{i+1}, \sum_{j\in\settwompp{2}}b_j - \sum_{j \in \setmpp} b_j + b_i \} 
\ge \sum_{j\in\settwompp{2}}b_j - \sum_{j \in \setmpp} b_j + b_i 
\end{eqnarray*}
%
Since bidders are conservative and losers bid their values,
\begin{eqnarray*}
p^{VCG}_{i}
    &=& \sum_{j \in \settwoopt{-i}} v_j - \sum_{j \in \setopt } v_j + v_i  \le \sum_{j\in \settwompp{2}}v_{j} - \sum_{j\in \setmpp}b_{j} +b_{i}\\
    &=& \sum_{j \in \anotb{\settwompp{2}}{\setmpp}} b_{j} + \sum_{j \in \settwompp{2} \cap \setmpp} b_{j} 
    - \sum_{j \in \setmpp \cap \settwompp{2}} b_{j}  - \sum_{j \in \anotb{\setmpp}{\settwompp{2}}} b_{j} +b_{i}
     = p^{MPP}_{i}
    \end{eqnarray*}
    
Now, consider payments of all such bidders  $i\in\anotb{\setopt}{\setmpp}$ or $i\in\anotb{\setmpp}{\setopt}$. This is possible, when there are 2 allocations of different size that have equally high efficiency, lets denote them by $\alloc_{\setopt}$ and $\alloc_{\setmpp}$. If bidder $i$ is present in only one of allocations, then her payment is $p_{i}^{MPP} = p_{i}^{VCG} = v_{i}$. Payment of $\VCG$ follows from definition.
 Observe, that bidder $i$ submits truthful bid in $\MPP$, because otherwise she will be not be in winning configuration. Now, one can derive the payment from definition.
\esa{
\hfill$\Box$
}{}
\end{proof}

With rational bidders,  we show that revenue of $\MPP$ can be as low as half of that of $\VCG$.

\medskip
\noindent 
\begin{example}
Consider 3 bidders with the following valuations $A=(100+\epsilon,1)$, $B=(50,2)$ and $C=(50, 2)$. Rational bidders can converge to bids $A=(100,1)$, $B=(100,2)$ and $C=(50,2)$ respectively. $\VCG$ gets truthful bids and chooses allocation consisting of bidder $A$ and her payment is 100, while $\MPP$ chooses allocation with bidders ($B$ and $C$) and prices them 50 and 0 respectively, achieving exactly half of revenue of $\VCG$.
\hfill $\Box$
\end{example}

\medskip
As in case with efficiency, revenue of cardinal auctions is also surprising in contrast with other auctions.  It is believed, that one of the reasons to use MPP auctions is to improve revenue, e.g., revenue of $GSP$ without click-through rates is always at least as much as that of $VCG$. This is also true for modification presented in~\cite{ghosh2010expressive}. 
Likewise, for prefix auctions, this continues to hold. 

\begin{theorem}
In equilibrium, $Rev(pVCG) \le Rev(pGSP)$. 
\end{theorem}

\begin{proof}
Let $\alloc$ be the allocation of $pGSP$ (or $pVCG$). Consider payment of bidder $i\in \alloc$. Let $l$ be the bidder who enters allocation $\alloc$ if $i$ leaves it. If no such bidder exist, let $l$ be a bidder with valuation $v_{l} = 0$ and $k_{l} = n$.
Then payment of bidder $i$ in $pVCG$ is 
$p_{i}^{pVCG} = \eff_{-i} + \eff + v_{i} = v_{l}$
and ayment of bidder $i$ in $pGSP$ allocation
$p_{i}^{pGSP} = \max\{b_{i+1}, b_{l}\}$.

Payment is minimized when $p_{i}^{pGSP} = b_{l}$. $b_{l} = v_{l}$, since bidders are conservative, and $l$ is loosing bidder. Hence, $Rev(pVCG) \le Rev(pGSP)$. 
\esa{
\hfill$\Box$
}{}
\end{proof}

In contrast to $pVCG$ and $VCG$ that have lower revenue than the corresponding versions of MPP, for cardinal auctions, we have shown that in some cases $\VCG$ may have more revenue than $\MPP$. This is because cardinal constraints enable richer strategies for bidders in particular strategy of rational bidders who can bid above their value.

\section{Concluding Remarks and Future Directions}
We consider the problem of selling identical copies of an item via an auction in which the number of copies sold is unknown \emph{a priori}, and valuation of a bidder depends on the total number of winners. This scenario is motivated by  number of ads on a page or number of parties that get access to certain information. 
While there are many ways to solve this problem, we consider \emph{cardinal auctions} in which the bidding language 
lets buyers explicitly bid on the maximum number of winners allowed. Our work analyzes cardinal auctions of $\MPP$ and $\VCG$ for revenue and efficiency tradeoffs in equilibria, and shows that they are quite different from the case without the cardinal externality. We find that $\MPP$ which is inspired by widely used Generalized Second Price (GSP) auction has surprising properties. In case of rational bidders efficiency of $\MPP$ is half of that of $\VCG$. At the same time, in the worst case $\MPP$ can collect only half of revenue of $\VCG$. 

There are many open directions to pursue. For example, in display ads, slots may differ 
in terms of their location and dimensions, as well as  click through rates. We need to extend the study of cardinal auctions to auctions for configurations of display ads with varying quality scores or with varying click through models.

Externality is a richer phenomenon than we have studied here. For instance, the value for a bidder
might depend not only on the number of other possessors, but also on their identity, quality, etc. 
Further, one can consider bidding languages which go beyond the step function we have adopted here,
for example, by letting bidders specify their value for each potential number of winners. 
Studying such notions of externalities and bidding languages is an active area in Economics and problems are still open. 

From a technical point of view, we would like to extend our analysis to Bayesian case and 
study dynamics of cardinal auctions. 

\bibliography{algorithmica}

\newcommand{\etalchar}[1]{$^{#1}$}
\begin{thebibliography}{GHK{\etalchar{+}}06}

\bibitem[AC08]{aseff2008optimal}
J.~Aseff and H.~Chade.
\newblock An optimal auction with identity-dependent externalities.
\newblock {\em The RAND Journal of Economics}, 39(3):731--746, 2008.

\bibitem[AFM07]{aggarwal2007bidding}
G.~Aggarwal, J.~Feldman, and S.~Muthukrishnan.
\newblock Bidding to the top: Vcg and equilibria of position-based auctions.
\newblock {\em Approximation and Online Algorithms}, pages 15--28, 2007.

\bibitem[AGM06]{aggarwal2006truthful}
G.~Aggarwal, A.~Goel, and R.~Motwani.
\newblock {Truthful auctions for pricing search keywords}.
\newblock In {\em Proceedings of the 7th ACM conference on Electronic
  commerce}, pages 1--7. ACM, 2006.

\bibitem[BCF{\etalchar{+}}08]{broder2008swing}
A.~Broder, M.~Ciaramita, M.~Fontoura, E.~Gabrilovich, V.~Josifovski,
  D.~Metzler, V.~Murdock, and V.~Plachouras.
\newblock {To swing or not to swing: Learning when (not) to advertise}.
\newblock In {\em Proceeding of the 17th ACM conference on Information and
  knowledge management}, pages 1003--1012. ACM, 2008.

\bibitem[Cla73]{G}
E.~Clarke.
\newblock {Multipart pricing of public goods}.
\newblock In {\em Public Choice}, page 41:617Ð631, 1973.

\bibitem[EOS05]{edelman2005internet}
B.~Edelman, M.~Ostrovsky, and M.~Schwarz.
\newblock {Internet advertising and the generalized second price auction:
  Selling billions of dollars worth of keywords}, 2005.

\bibitem[GH03]{goldberg2003envy}
A.V. Goldberg and J.D. Hartline.
\newblock Envy-free auctions for digital goods.
\newblock In {\em Proceedings of the 4th ACM conference on Electronic
  commerce}, pages 29--35. ACM, 2003.

\bibitem[GHK{\etalchar{+}}06]{goldberg2006comp}
A.V. Goldberg, J.D. Hartline, A.R. Karlin, M.~Saks, and A.~Wright.
\newblock Competitive auctions.
\newblock {\em Games and Economic Behavior}, 55(2):242--269, 2006.

\bibitem[Gro71]{V}
T.~Groves.
\newblock {Incentives in teams}.
\newblock In {\em Econometrica,}, pages 11:17--33, 1971.

\bibitem[GS10]{ghosh2010expressive}
A.~Ghosh and A.~Sayedi.
\newblock {Expressive auctions for externalities in online advertising}.
\newblock In {\em Proceedings of the 19th international conference on World
  wide web}, pages 371--380. ACM, 2010.

\bibitem[HIMM11]{haghpanah2011optimal}
N.~Haghpanah, N.~Immorlica, VS~Mirrokni, and K.~Munagala.
\newblock Optimal auctions with positive network externalities.
\newblock In {\em Proc. of the 12th ACM Conference on Electronic Commerce
  (ECÕ11)}, pages 11--20, 2011.

\bibitem[JMS96]{jehiel1996not}
P.~Jehiel, B.~Moldovanu, and E.~Stacchetti.
\newblock How (not) to sell nuclear weapons.
\newblock {\em The American Economic Review}, pages 814--829, 1996.

\bibitem[JS11]{jerath2011exclusive}
K.~Jerath and A.~Sayedi.
\newblock Exclusive display in sponsored search advertising.
\newblock 2011.

\bibitem[LPL11]{lucier2011gsp}
B.~Lucier and R.~Paes~Leme.
\newblock Gsp auctions with correlated types.
\newblock In {\em Proceedings of the 12th ACM conference on Electronic
  commerce}, pages 71--80. ACM, 2011.

\bibitem[LT10]{leme2010pure}
R.P. Leme and E.~Tardos.
\newblock {Pure and Bayes-Nash price of anarchy for generalized second price
  auction}.
\newblock In {\em 2010 IEEE 51st Annual Symposium on Foundations of Computer
  Science}, pages 735--744. IEEE, 2010.

\bibitem[Mut09]{muthukrishnan2009bidding}
S.~Muthukrishnan.
\newblock Bidding on configurations in internet ad auctions.
\newblock {\em Computing and Combinatorics}, pages 1--6, 2009.

\bibitem[PKK11]{pei2011working}
J.~Pei, D.~Klabjan, and F.~Karaesmen.
\newblock {Auction Analysis for Sharable Goods with Externality}.
\newblock {\em Working Paper}, 2011.

\bibitem[SKN{\etalchar{+}}10]{schroedl2010generalized}
S.~Schroedl, A.~Kesari, A.~Nair, L.~Neumeyer, and S.~Rao.
\newblock {Generalized Utility in Web Search Advertising}.
\newblock 2010.

\bibitem[SS71]{shapley1971assignment}
L.S. Shapley and M.~Shubik.
\newblock The assignment game i: The core.
\newblock {\em International Journal of Game Theory}, 1(1):111--130, 1971.

\bibitem[Vic61]{C}
W.~Vickrey.
\newblock {Counterspeculation, auctions and competitive sealed tenders}.
\newblock In {\em Journal of Finance}, page 16:8Ð37, 1961.

\end{thebibliography}
\bibliographystyle{alpha}

 \end{document}